\documentclass[onecolumn,a4paper,11pt]{article}

\usepackage[utf8]{inputenc} 
\usepackage{textcomp} 
\usepackage{graphicx}  
\usepackage{flafter}  
\usepackage[normalem]{ulem}

\usepackage{amsmath,amsthm}
\usepackage{bm}  
\usepackage{bbm}
\usepackage{colonequals}

\usepackage[bitstream-charter]{mathdesign}
\usepackage[T1]{fontenc}

\usepackage[margin=2cm,a4paper]{geometry}
\usepackage{verbatim}
\usepackage{xcolor}
\usepackage{units}

\usepackage[bookmarks,colorlinks,breaklinks]{hyperref}  
\definecolor{dullmagenta}{rgb}{0.4,0,0.4}   
\definecolor{darkblue}{rgb}{0,0,0.4}
\hypersetup{linkcolor=red,citecolor=blue,filecolor=dullmagenta,urlcolor=darkblue} 
\usepackage{pdfsync}  

\usepackage[backend=bibtex,sorting=none,style=phys,biblabel=brackets,backref=true,bibencoding=utf8,maxnames=20]{biblatex}
\DefineBibliographyStrings{english}{%
  backrefpage = {page},
  backrefpages = {pages},
}
\makeatletter
\def\blx@bblfile@bibtex{%
  \blx@secinit
  \begingroup
  \blx@bblstart
\begingroup
\makeatletter
\@ifundefined{ver@biblatex.sty}
  {\@latex@error
     {Missing 'biblatex' package}
     {The bibliography requires the 'biblatex' package.}
      \aftergroup }
  {}
\endgroup

\entry{tomamichel_quantum_2015v1}{article}{}
  \name{author}{2}{}{%
    {{}%
     {Tomamichel}{T.}%
     {Marco}{M.}%
     {}{}%
     {}{}}%
    {{}%
     {Berta}{B.}%
     {Mario}{M.}%
     {}{}%
     {}{}}%
  }
  \strng{namehash}{TMBM1}
  \strng{fullhash}{TMBM1}
  \field{abstract}{%
  The quantum capacity of a channel is often used as a single measure to
  characterize the ability of a channel to transmit quantum information
  coherently. The capacity determines the maximal asymptotic rate at which we
  can code reliably over a channel. Here we argue that this asymptotic
  treatment is insufficient to the point of being irrelevant in the quantum
  setting when decoherence limits our ability to manipulate large quantum
  systems. For all practical purposes we should instead focus on the tradeoff
  between three parameters: the rate of the code, the number of coherent uses
  of the channel, and the fidelity of the transmission. The aim is then to
  specify the region determined by allowed combinations of these parameters.
  Towards this goal, we find an approximate characterization of the region of
  allowed triplets for the qubit dephasing channel and for the erasure channel
  with classical post-processing. In each case the region is parametrized by a
  second channel parameter, the dispersion of the quantum capacity. In the
  process we also develop general bounds on the achievable region that are
  valid for all quantum channels.%
  }
  \field{title}{Quantum {Coding} over {Dephasing} and {Erasure} {Channels} with
  {Finite} {Resources}}
  \verb{url}
  \verb http://arxiv.org/abs/1504.04617v1
  \endverb
  \verb{file}
  \verb Tomamichel_Berta_2015_Quantum Coding over Dephasing and Erasure Channel
  \verb s with Finite Resources.pdf:/Users/joe/store/zotero/storage/3KGMTJFA/To
  \verb mamichel_Berta_2015_Quantum Coding over Dephasing and Erasure Channels
  \verb with Finite Resources.pdf:application/pdf
  \endverb
  \field{journaltitle}{arXiv:1504.04617v1 [quant-ph]}
  \field{month}{04}
  \field{year}{2015}
  \field{urlday}{20}
  \field{urlmonth}{04}
  \field{urlyear}{2015}
\endentry

\entry{tomamichel_quantum_2015}{article}{}
  \name{author}{3}{}{%
    {{}%
     {Tomamichel}{T.}%
     {Marco}{M.}%
     {}{}%
     {}{}}%
    {{}%
     {Berta}{B.}%
     {Mario}{M.}%
     {}{}%
     {}{}}%
    {{}%
     {Renes}{R.}%
     {Joseph~M.}{J.~M.}%
     {}{}%
     {}{}}%
  }
  \strng{namehash}{TMBMRJM1}
  \strng{fullhash}{TMBMRJM1}
  \field{abstract}{%
  The quantum capacity of a memoryless channel is often used as a single figure
  of merit to characterize its ability to transmit quantum information
  coherently. The capacity determines the maximal rate at which we can code
  reliably over asymptotically many uses of the channel. We argue that this
  asymptotic treatment is insufficient to the point of being irrelevant in the
  quantum setting where decoherence severely limits our ability to manipulate
  large quantum systems in the encoder and decoder. For all practical purposes
  we should instead focus on the trade-off between three parameters: the rate
  of the code, the number of coherent uses of the channel, and the fidelity of
  the transmission. The aim is then to specify the region determined by allowed
  combinations of these parameters. Towards this goal, we find approximate and
  exact characterizations of the region of allowed triplets for the qubit
  dephasing channel and for the erasure channel with classical post-processing
  assistance. In each case the region is parametrized by a second channel
  parameter, the quantum channel dispersion. In the process we also develop
  several general inner (achievable) and outer (converse) bounds on the coding
  region that are valid for all finite-dimensional quantum channels and can be
  computed efficiently. Applied to the depolarizing channel, this allows us to
  determine a lower bound on the number of coherent uses of the channel
  necessary to witness super-additivity of the coherent information.%
  }
  \field{title}{Quantum {Coding} with {Finite} {Resources}}
  \verb{url}
  \verb http://arxiv.org/abs/1504.04617v2
  \endverb
  \verb{file}
  \verb Tomamichel et al_2015_Quantum Coding with Finite Resources.pdf:/Users/j
  \verb oe/store/zotero/storage/WFZ9JRAE/Tomamichel et al_2015_Quantum Coding w
  \verb ith Finite Resources.pdf:application/pdf
  \endverb
  \field{journaltitle}{arXiv:1504.04617v2 [quant-ph]}
  \field{month}{04}
  \field{year}{2015}
  \field{urlday}{29}
  \field{urlmonth}{05}
  \field{urlyear}{2015}
\endentry

\entry{arikan_channel_2009}{article}{}
  \name{author}{1}{}{%
    {{}%
     {Ar{\i}kan}{A.}%
     {E.}{E.}%
     {}{}%
     {}{}}%
  }
  \keyw{binary codes, Capacity-achieving codes, Capacity planning, channel
  capacity, channel coding, channel polarization, codes, code sequence,
  Councils, decoding, information theory, memoryless systems, Noise
  cancellation, Plotkin construction, polar codes, Polarization, probability,
  Reed– Muller (RM) codes, successive cancellation decoding, successive
  cancellation decoding algorithm, symmetric binary-input memoryless channel}
  \strng{namehash}{AE1}
  \strng{fullhash}{AE1}
  \field{abstract}{%
  A method is proposed, called channel polarization, to construct code
  sequences that achieve the symmetric capacity I(W) of any given binary-input
  discrete memoryless channel (B-DMC) W. The symmetric capacity is the highest
  rate achievable subject to using the input letters of the channel with equal
  probability. Channel polarization refers to the fact that it is possible to
  synthesize, out of N independent copies of a given B-DMC W, a second set of N
  binary-input channels \{WN(i)1 les i les N\} such that, as N becomes large,
  the fraction of indices i for which I(WN(i)) is near 1 approaches I(W) and
  the fraction for which I(WN(i)) is near 0 approaches 1-I(W). The polarized
  channels \{WN(i)\} are well-conditioned for channel coding: one need only
  send data at rate 1 through those with capacity near 1 and at rate 0 through
  the remaining. Codes constructed on the basis of this idea are called polar
  codes. The paper proves that, given any B-DMC W with I(W) {\textgreater} 0
  and any target rate R{\textless} I(W) there exists a sequence of polar codes
  \{Cfrn;nges1\} such that Cfrn has block-length N=2n , rate ges R, and
  probability of block error under successive cancellation decoding bounded as
  Pe(N,R) les O(N-1/4) independently of the code rate. This performance is
  achievable by encoders and decoders with complexity O(N logN) for each.%
  }
  \verb{doi}
  \verb 10.1109/TIT.2009.2021379
  \endverb
  \field{issn}{0018-9448}
  \field{number}{7}
  \field{pages}{3051\bibrangedash 3073}
  \field{shorttitle}{Channel {Polarization}}
  \field{title}{Channel {Polarization}: {A} {Method} for {Constructing}
  {Capacity}-{Achieving} {Codes} for {Symmetric} {Binary}-{Input} {Memoryless}
  {Channels}}
  \verb{url}
  \verb http://dx.doi.org/10.1109/TIT.2009.2021379
  \endverb
  \field{volume}{55}
  \verb{file}
  \verb Arikan - 2009 - Channel Polarization A Method for Constructing Ca.pdf:/
  \verb Users/joe/store/zotero/storage/9KWIT886/Arikan - 2009 - Channel Polariz
  \verb ation A Method for Constructing Ca.pdf:application/pdf
  \endverb
  \field{journaltitle}{IEEE Transactions on Information Theory}
  \field{month}{07}
  \field{year}{2009}
\endentry

\entry{kudekar_spatially_2012}{article}{}
  \name{author}{3}{}{%
    {{}%
     {Kudekar}{K.}%
     {Shrinivas}{S.}%
     {}{}%
     {}{}}%
    {{}%
     {Richardson}{R.}%
     {Tom}{T.}%
     {}{}%
     {}{}}%
    {{}%
     {Urbanke}{U.}%
     {Ruediger}{R.}%
     {}{}%
     {}{}}%
  }
  \keyw{Computer Science - Information Theory}
  \strng{namehash}{KSRTUR1}
  \strng{fullhash}{KSRTUR1}
  \field{abstract}{%
  We investigate spatially coupled code ensembles. For transmission over the
  binary erasure channel, it was recently shown that spatial coupling increases
  the belief propagation threshold of the ensemble to essentially the maximum
  a-priori threshold of the underlying component ensemble. This explains why
  convolutional LDPC ensembles, originally introduced by Felstrom and
  Zigangirov, perform so well over this channel. We show that the equivalent
  result holds true for transmission over general binary-input memoryless
  output-symmetric channels. More precisely, given a desired error probability
  and a gap to capacity, we can construct a spatially coupled ensemble which
  fulfills these constraints universally on this class of channels under belief
  propagation decoding. In fact, most codes in that ensemble have that
  property. The quantifier universal refers to the single ensemble/code which
  is good for all channels but we assume that the channel is known at the
  receiver. The key technical result is a proof that under belief propagation
  decoding spatially coupled ensembles achieve essentially the area threshold
  of the underlying uncoupled ensemble. We conclude by discussing some
  interesting open problems.%
  }
  \field{title}{Spatially {Coupled} {Ensembles} {Universally} {Achieve}
  {Capacity} under {Belief} {Propagation}}
  \verb{url}
  \verb http://arxiv.org/abs/1201.2999
  \endverb
  \verb{file}
  \verb 1201.2999 PDF:/Users/joe/store/zotero/storage/TNVNCKAU/Kudekar et al. -
  \verb  2012 - Spatially Coupled Ensembles Universally Achieve Ca.pdf:applicat
  \verb ion/pdf
  \endverb
  \field{journaltitle}{arXiv:1201.2999 [cs.IT]}
  \field{month}{01}
  \field{year}{2012}
  \field{urlday}{06}
  \field{urlmonth}{09}
  \field{urlyear}{2012}
\endentry

\entry{polyanskiy_channel_2010}{article}{}
  \name{author}{3}{}{%
    {{}%
     {Polyanskiy}{P.}%
     {Y.}{Y.}%
     {}{}%
     {}{}}%
    {{}%
     {Poor}{P.}%
     {H.V.}{H.}%
     {}{}%
     {}{}}%
    {{}%
     {Verd\'u}{V.}%
     {S.}{S.}%
     {}{}%
     {}{}}%
  }
  \keyw{Achievability, Acoustic noise, AWGN, Capacity planning, channel
  capacity, channel coding, channel dispersion, codes, coding for noisy
  channels, complementary Gaussian cumulative distribution function, converse,
  Distribution functions, error probability, error statistics, finite
  blocklength regime, Helium, maximal channel coding rate, Shannon theory,
  Upper bound}
  \strng{namehash}{PYPHVS1}
  \strng{fullhash}{PYPHVS1}
  \field{abstract}{%
  This paper investigates the maximal channel coding rate achievable at a given
  blocklength and error probability. For general classes of channels new
  achievability and converse bounds are given, which are tighter than existing
  bounds for wide ranges of parameters of interest, and lead to tight
  approximations of the maximal achievable rate for blocklengths n as short as
  100. It is also shown analytically that the maximal rate achievable with
  error probability Â¿ isclosely approximated by C - Â¿(V/n) Q-1(Â¿)
  where C is the capacity, V is a characteristic of the channel referred to as
  channel dispersion , and Q is the complementary Gaussian cumulative
  distribution function.%
  }
  \verb{doi}
  \verb 10.1109/TIT.2010.2043769
  \endverb
  \field{issn}{0018-9448}
  \field{number}{5}
  \field{pages}{2307\bibrangedash 2359}
  \field{title}{Channel {Coding} {Rate} in the {Finite} {Blocklength} {Regime}}
  \field{volume}{56}
  \verb{file}
  \verb Polyanskiy et al. - 2010 - Channel Coding Rate in the Finite Blocklengt
  \verb h Regi.pdf:/Users/joe/store/zotero/storage/RSP2ARW2/Polyanskiy et al. -
  \verb  2010 - Channel Coding Rate in the Finite Blocklength Regi.pdf:applicat
  \verb ion/pdf
  \endverb
  \field{journaltitle}{IEEE Transactions on Information Theory}
  \field{year}{2010}
\endentry

\entry{matthews_linear_2012}{article}{}
  \name{author}{1}{}{%
    {{}%
     {Matthews}{M.}%
     {W.}{W.}%
     {}{}%
     {}{}}%
  }
  \strng{namehash}{MW1}
  \strng{fullhash}{MW1}
  \verb{doi}
  \verb 10.1109/TIT.2012.2210695
  \endverb
  \field{issn}{0018-9448, 1557-9654}
  \field{number}{12}
  \field{pages}{7036\bibrangedash 7044}
  \field{title}{A {Linear} {Program} for the {Finite} {Block} {Length}
  {Converse} of {Polyanskiy}-{Poor}-{Verdú} {Via} {Nonsignaling} {Codes}}
  \verb{url}
  \verb http://ieeexplore.ieee.org/lpdocs/epic03/wrapper.htm?arnumber=6269084
  \endverb
  \field{volume}{58}
  \verb{file}
  \verb Matthews_2012_A Linear Program for the Finite Block Length Converse of
  \verb Polyanskiy-Poor-Verdú.pdf:/Users/joe/store/zotero/storage/A8H67TK9/Mat
  \verb thews_2012_A Linear Program for the Finite Block Length Converse of Pol
  \verb yanskiy-Poor-Verdú.pdf:application/pdf
  \endverb
  \field{journaltitle}{IEEE Transactions on Information Theory}
  \field{month}{12}
  \field{year}{2012}
  \field{urlday}{21}
  \field{urlmonth}{05}
  \field{urlyear}{2014}
\endentry

\entry{polyanskiy_saddle_2013}{article}{}
  \name{author}{1}{}{%
    {{}%
     {Polyanskiy}{P.}%
     {Y.}{Y.}%
     {}{}%
     {}{}}%
  }
  \keyw{binary erasure channel, Binary hypothesis testing, channel coding,
  channel symmetries, channel symmetry group, converse theorem, decoding, error
  exponents, Extraterrestrial measurements, information-spectrum converse
  bound, input distribution, input distributions, memoryless systems, minimax,
  minimax converse, minimax metaconverse, minimax techniques, nonasymptotic
  analysis, optimal output distributions, Optimization, optimization problem,
  saddle point, Shannon theory, spectrum bound, Standards, symmetry
  considerations, Testing, Verdú-Han converse bound}
  \strng{namehash}{PY1}
  \strng{fullhash}{PY1}
  \field{abstract}{%
  A minimax metaconverse has recently been proposed as a simultaneous
  generalization of a number of classical results and a tool for the
  nonasymptotic analysis. In this paper, it is shown that the order of
  optimizing the input and output distributions can be interchanged without
  affecting the bound. In the course of the proof, a number of auxiliary
  results of separate interest are obtained. In particular, it is shown that
  the optimization problem is convex and can be solved in many cases by the
  symmetry considerations. As a consequence, it is demonstrated that in the
  latter cases, the (multiletter) input distribution in information-spectrum
  (Verdú-Han) converse bound can be taken to be a (memoryless) product of
  single-letter ones. A tight converse for the binary erasure channel is
  rederived by computing the optimal (nonproduct) output distribution. For
  discrete memoryless channels, a conjecture of Poor and Verdú regarding the
  tightness of the information spectrum bound on the error exponents is
  resolved in the negative. Concept of the channel symmetry group is
  established and relations with the definitions of symmetry by Gallager and
  Dobrushin are investigated.%
  }
  \verb{doi}
  \verb 10.1109/TIT.2012.2236382
  \endverb
  \field{issn}{0018-9448}
  \field{number}{5}
  \field{pages}{2576\bibrangedash 2595}
  \field{title}{Saddle {Point} in the {Minimax} {Converse} for {Channel}
  {Coding}}
  \field{volume}{59}
  \verb{file}
  \verb Polyanskiy - 2013 - Saddle Point in the Minimax Converse for Channel C.
  \verb pdf:/Users/joe/store/zotero/storage/I476TNZM/Polyanskiy - 2013 - Saddle
  \verb  Point in the Minimax Converse for Channel C.pdf:application/pdf
  \endverb
  \field{journaltitle}{IEEE Transactions on Information Theory}
  \field{year}{2013}
\endentry

\entry{matthews_finite_2014}{article}{}
  \name{author}{2}{}{%
    {{}%
     {Matthews}{M.}%
     {William}{W.}%
     {}{}%
     {}{}}%
    {{}%
     {Wehner}{W.}%
     {Stephanie}{S.}%
     {}{}%
     {}{}}%
  }
  \strng{namehash}{MWWS1}
  \strng{fullhash}{MWWS1}
  \verb{doi}
  \verb 10.1109/TIT.2014.2353614
  \endverb
  \field{issn}{0018-9448, 1557-9654}
  \field{number}{11}
  \field{pages}{7317\bibrangedash 7329}
  \field{title}{Finite {Blocklength} {Converse} {Bounds} for {Quantum}
  {Channels}}
  \verb{url}
  \verb http://ieeexplore.ieee.org/lpdocs/epic03/wrapper.htm?arnumber=6891222
  \endverb
  \field{volume}{60}
  \verb{file}
  \verb Matthews_Wehner_2014_Finite Blocklength Converse Bounds for Quantum Cha
  \verb nnels.pdf:/Users/joe/store/zotero/storage/XMIG5P5V/Matthews_Wehner_2014
  \verb _Finite Blocklength Converse Bounds for Quantum Channels.pdf:applicatio
  \verb n/pdf
  \endverb
  \field{journaltitle}{IEEE Transactions on Information Theory}
  \field{month}{11}
  \field{year}{2014}
  \field{urlday}{20}
  \field{urlmonth}{04}
  \field{urlyear}{2015}
\endentry

\entry{leung_power_2014}{article}{}
  \name{author}{2}{}{%
    {{}%
     {Leung}{L.}%
     {Debbie}{D.}%
     {}{}%
     {}{}}%
    {{}%
     {Matthews}{M.}%
     {William}{W.}%
     {}{}%
     {}{}}%
  }
  \strng{namehash}{LDMW1}
  \strng{fullhash}{LDMW1}
  \field{abstract}{%
  We derive one-shot upper bounds for quantum noisy channel codes. We do so by
  regarding a channel code as a bipartite operation with an encoder belonging
  to the sender and a decoder belonging to the receiver, and imposing
  constraints on the bipartite operation. We investigate the power of codes
  whose bipartite operation is non-signalling from Alice to Bob,
  positive-partial transpose (PPT) preserving, or both, and derive a simple
  semidefinite program for the achievable entanglement fidelity. Using the
  semidefinite program, we show that the non-signalling assisted quantum
  capacity for memoryless channels is equal to the entanglement-assisted
  capacity. We also relate our PPT-preserving codes and the PPT-preserving
  entanglement distillation protocols studied by Rains. Applying these results
  to a concrete example, the 3-dimensional Werner-Holevo channel, we find that
  codes that are non-signalling and PPT-preserving can be strictly less
  powerful than codes satisfying either one of the constraints, and therefore
  provide a tighter bound for unassisted codes. Furthermore, PPT-preserving
  non-signalling codes can send one qubit perfectly over two uses of the
  channel, which has no quantum capacity. We discuss whether this can be
  interpreted as a form of superactivation of quantum capacity.%
  }
  \field{title}{On the power of {PPT}-preserving and non-signalling codes}
  \verb{url}
  \verb http://arxiv.org/abs/1406.7142
  \endverb
  \verb{file}
  \verb Leung_Matthews_2014_On the power of PPT-preserving and non-signalling c
  \verb odes.pdf:/Users/joe/store/zotero/storage/9R9J4RAH/Leung_Matthews_2014_O
  \verb n the power of PPT-preserving and non-signalling codes.pdf:application/
  \verb pdf
  \endverb
  \field{journaltitle}{arXiv:1406.7142 [quant-ph]}
  \field{month}{06}
  \field{year}{2014}
  \field{urlday}{12}
  \field{urlmonth}{02}
  \field{urlyear}{2015}
\endentry

\entry{choi_completely_1975}{article}{}
  \name{author}{1}{}{%
    {{}%
     {Choi}{C.}%
     {Man-Duen}{M.-D.}%
     {}{}%
     {}{}}%
  }
  \strng{namehash}{CMD1}
  \strng{fullhash}{CMD1}
  \field{abstract}{%
  A linear map Φ from Mn to Mm is completely positive iff it admits an
  expression Φ(A)=ΣiV∗iAVi where Vi are n×m matrices.%
  }
  \verb{doi}
  \verb 10.1016/0024-3795(75)90075-0
  \endverb
  \field{issn}{0024-3795}
  \field{number}{3}
  \field{pages}{285\bibrangedash 290}
  \field{title}{Completely positive linear maps on complex matrices}
  \verb{url}
  \verb http://www.sciencedirect.com/science/article/pii/0024379575900750
  \endverb
  \field{volume}{10}
  \verb{file}
  \verb Choi - 1975 - Completely positive linear maps on complex matrice.pdf:/U
  \verb sers/joe/store/zotero/storage/JZWTF2CH/Choi - 1975 - Completely positiv
  \verb e linear maps on complex matrice.pdf:application/pdf
  \endverb
  \field{journaltitle}{Linear Algebra and its Applications}
  \field{month}{06}
  \field{year}{1975}
  \field{urlday}{30}
  \field{urlmonth}{08}
  \field{urlyear}{2013}
\endentry

\entry{jamiolkowski_linear_1972}{article}{}
  \name{author}{1}{}{%
    {{}%
     {Jamiołkowski}{J.}%
     {A.}{A.}%
     {}{}%
     {}{}}%
  }
  \strng{namehash}{JA1}
  \strng{fullhash}{JA1}
  \field{abstract}{%
  This work may be considered a completion of the paper by J. de Pillis: Linear
  transformations which preserve Hermitian and positive semidefinite operators,
  published in 1967 [2]: necessary conditions have been formulated. Let A1 be
  the full algebra of linear operators on the n-dimensional Hilbert space H1,
  and let A2 be the full algebra of linear operators on the m-dimensional
  Hilbert space H2. Let L(A1,A2) denote the complex spaceof linear maps from A1
  to A2 and S denotes the cone of all T ϵ L(A1,A2 which send positive
  semidefinite operators from A1 to positive semidefinite operators from A2.
  The aim of this paper is to present a necessary and sufficient condition for
  a transformation in L(A1, A2) to be in the cone S, and to preserve trace of
  the operators.%
  }
  \verb{doi}
  \verb 10.1016/0034-4877(72)90011-0
  \endverb
  \field{issn}{0034-4877}
  \field{number}{4}
  \field{pages}{275\bibrangedash 278}
  \field{title}{Linear transformations which preserve trace and positive
  semidefiniteness of operators}
  \verb{url}
  \verb http://www.sciencedirect.com/science/article/pii/0034487772900110
  \endverb
  \field{volume}{3}
  \verb{file}
  \verb Jamiołkowski - 1972 - Linear transformations which preserve trace and
  \verb po.pdf:/Users/joe/store/zotero/storage/4G5636UF/Jamiołkowski - 1972 -
  \verb Linear transformations which preserve trace and po.pdf:application/pdf
  \endverb
  \field{journaltitle}{Reports on Mathematical Physics}
  \field{month}{12}
  \field{year}{1972}
  \field{urlday}{30}
  \field{urlmonth}{08}
  \field{urlyear}{2013}
\endentry

\entry{leifer_towards_2013}{article}{}
  \name{author}{2}{}{%
    {{}%
     {Leifer}{L.}%
     {M.~S.}{M.~S.}%
     {}{}%
     {}{}}%
    {{}%
     {Spekkens}{S.}%
     {Robert~W.}{R.~W.}%
     {}{}%
     {}{}}%
  }
  \strng{namehash}{LMSSRW1}
  \strng{fullhash}{LMSSRW1}
  \field{abstract}{%
  Quantum theory can be viewed as a generalization of classical probability
  theory, but the analogy as it has been developed so far is not complete.
  Whereas the manner in which inferences are made in classical probability
  theory is independent of the causal relation that holds between the
  conditioned variable and the conditioning variable, in the conventional
  quantum formalism, there is a significant difference between how one treats
  experiments involving two systems at a single time and those involving a
  single system at two times. In this article, we develop the formalism of
  quantum conditional states, which provides a unified description of these two
  sorts of experiment. In addition, concepts that are distinct in the
  conventional formalism become unified: Channels, sets of states, and positive
  operator valued measures are all seen to be instances of conditional states;
  the action of a channel on a state, ensemble averaging, the Born rule, the
  composition of channels, and nonselective state-update rules are all seen to
  be instances of belief propagation. Using a quantum generalization of
  Bayes’ theorem and the associated notion of Bayesian conditioning, we also
  show that the remote steering of quantum states can be described within our
  formalism as a mere updating of beliefs about one system given new
  information about another, and retrodictive inferences can be expressed using
  the same belief propagation rule as is used for predictive inferences.
  Finally, we show that previous arguments for interpreting the projection
  postulate as a quantum generalization of Bayesian conditioning are based on a
  misleading analogy and that it is best understood as a combination of belief
  propagation (corresponding to the nonselective state-update map) and
  conditioning on the measurement outcome.%
  }
  \verb{doi}
  \verb 10.1103/PhysRevA.88.052130
  \endverb
  \field{number}{5}
  \field{pages}{052130}
  \field{title}{Towards a formulation of quantum theory as a causally neutral
  theory of {Bayesian} inference}
  \verb{url}
  \verb http://link.aps.org/doi/10.1103/PhysRevA.88.052130
  \endverb
  \field{volume}{88}
  \verb{file}
  \verb Leifer and Spekkens - 2013 - Towards a formulation of quantum theory as
  \verb  a causa.pdf:/Users/joe/store/zotero/storage/73XWVRJP/Leifer and Spekke
  \verb ns - 2013 - Towards a formulation of quantum theory as a causa.pdf:appl
  \verb ication/pdf
  \endverb
  \field{journaltitle}{Physical Review A}
  \field{month}{11}
  \field{year}{2013}
  \field{urlday}{05}
  \field{urlmonth}{03}
  \field{urlyear}{2015}
\endentry

\entry{rains_semidefinite_2001}{article}{}
  \name{author}{1}{}{%
    {{}%
     {Rains}{R.}%
     {E.~M}{E.~M}%
     {}{}%
     {}{}}%
  }
  \keyw{1-local distillation, asymmetric Werner states, clones, codes,
  distillable entanglement, hashing lower bound, Hilbert space, information
  theory, isotropic states, linear program, Linear programming, maximally
  correlated states, maximum fidelity, positive partial transpose distillation
  protocol, Production, protocols, quantum codes, quantum communication,
  quantum entanglement, Quantum mechanics, Rain, Refining, semidefinite
  program, symmetry, Upper bound, upper bounds}
  \strng{namehash}{REM1}
  \strng{fullhash}{REM1}
  \field{abstract}{%
  We show that the maximum fidelity obtained by a positive partial transpose
  (p.p.t.) distillation protocol is given by the solution to a certain
  semidefinite program. This gives a number of new lower and upper bounds on
  p.p.t. distillable entanglement (and thus new upper bounds on 2-locally
  distillable entanglement). In the presence of symmetry, the semidefinite
  program simplifies considerably, becoming a linear program in the case of
  isotropic and Werner states. Using these techniques, we determine the p.p.t.
  distillable entanglement of asymmetric Werner states and “maximally
  correlated” states. We conclude with a discussion of possible applications
  of semidefinite programming to quantum codes and 1-local distillation%
  }
  \verb{doi}
  \verb 10.1109/18.959270
  \endverb
  \field{issn}{0018-9448}
  \field{number}{7}
  \field{pages}{2921\bibrangedash 2933}
  \field{title}{A semidefinite program for distillable entanglement}
  \field{volume}{47}
  \verb{file}
  \verb Rains - 2001 - A semidefinite program for distillable entanglemen.pdf:/
  \verb Users/joe/store/zotero/storage/97AXJ2TB/Rains - 2001 - A semidefinite p
  \verb rogram for distillable entanglemen.pdf:application/pdf
  \endverb
  \field{journaltitle}{IEEE Transactions on Information Theory}
  \field{month}{11}
  \field{year}{2001}
\endentry

\entry{boyd_convex_2004}{book}{}
  \name{author}{2}{}{%
    {{}%
     {Boyd}{B.}%
     {Stephen}{S.}%
     {}{}%
     {}{}}%
    {{}%
     {Vandenberghe}{V.}%
     {Lieven}{L.}%
     {}{}%
     {}{}}%
  }
  \list{publisher}{1}{%
    {Cambridge University Press}%
  }
  \strng{namehash}{BSVL1}
  \strng{fullhash}{BSVL1}
  \field{isbn}{0521833787}
  \field{title}{Convex {Optimization}}
  \verb{file}
  \verb Boyd and Vandenberghe - 2004 - Convex Optimization.pdf:/Users/joe/store
  \verb /zotero/storage/G3G8QEWR/Boyd and Vandenberghe - 2004 - Convex Optimiza
  \verb tion.pdf:application/pdf
  \endverb
  \field{year}{2004}
\endentry

\entry{watrous_quantum_2011}{article}{}
  \name{author}{1}{}{%
    {{}%
     {Watrous}{W.}%
     {John}{J.}%
     {}{}%
     {}{}}%
  }
  \strng{namehash}{WJ1}
  \strng{fullhash}{WJ1}
  \field{title}{Quantum {Information} {Theory}}
  \verb{url}
  \verb https://cs.uwaterloo.ca/~watrous/LectureNotes.html
  \endverb
  \verb{file}
  \verb Watrous - 2011 - Quantum Information Theory.pdf:/Users/joe/store/zotero
  \verb /storage/N4Z9SMAR/Watrous - 2011 - Quantum Information Theory.pdf:appli
  \verb cation/pdf
  \endverb
  \field{type}{Lecture {Notes}}
  \field{journaltitle}{Lecture notes}
  \field{year}{2011}
\endentry

\entry{rains_bound_1999}{article}{}
  \name{author}{1}{}{%
    {{}%
     {Rains}{R.}%
     {E.~M.}{E.~M.}%
     {}{}%
     {}{}}%
  }
  \strng{namehash}{REM2}
  \strng{fullhash}{REM2}
  \field{abstract}{%
  The best bound known on 2-locally distillable entanglement is that of Vedral
  and Plenio, involving a certain measure of entanglement based on relative
  entropy. It turns out that a related argument can be used to give an even
  stronger bound; we give this bound, and examine some of its properties. In
  particular, and in contrast to the earlier bounds, the new bound is not
  additive in general. We give an example of a state for which the bound fails
  to be additive, as well as a number of states for which the bound is
  additive.%
  }
  \verb{doi}
  \verb 10.1103/PhysRevA.60.179
  \endverb
  \field{number}{1}
  \field{pages}{179}
  \field{title}{Bound on distillable entanglement}
  \verb{url}
  \verb http://link.aps.org/abstract/PRA/v60/p179
  \endverb
  \field{volume}{60}
  \verb{file}
  \verb Rains - 1999 - Bound on distillable entanglement.pdf:/Users/joe/store/z
  \verb otero/storage/E8HI42J9/Rains - 1999 - Bound on distillable entanglement
  \verb .pdf:application/pdf;Rains - 1999 - Erratum Bound on distillable entang
  \verb lement.pdf:/Users/joe/store/zotero/storage/4BS5V9I7/Rains - 1999 - Erra
  \verb tum Bound on distillable entanglement.pdf:application/pdf
  \endverb
  \field{journaltitle}{Physical Review A}
  \field{month}{07}
  \field{year}{1999}
  \field{urlday}{12}
  \field{urlmonth}{09}
  \field{urlyear}{2008}
\endentry

\entry{neumann_zur_1928}{article}{}
  \name{author}{1}{}{%
    {{}%
     {Neumann}{N.}%
     {J.~v}{J.~v}%
     {}{}%
     {}{}}%
  }
  \strng{namehash}{NJv1}
  \strng{fullhash}{NJv1}
  \verb{doi}
  \verb 10.1007/BF01448847
  \endverb
  \field{issn}{0025-5831, 1432-1807}
  \field{number}{1}
  \field{pages}{295\bibrangedash 320}
  \field{title}{Zur {Theorie} der {Gesellschaftsspiele}}
  \verb{url}
  \verb http://link.springer.com/article/10.1007/BF01448847
  \endverb
  \field{volume}{100}
  \verb{file}
  \verb Neumann - 1928 - Zur Theorie der Gesellschaftsspiele.pdf:/Users/joe/sto
  \verb re/zotero/storage/DFCEVC9I/Neumann - 1928 - Zur Theorie der Gesellschaf
  \verb tsspiele.pdf:application/pdf
  \endverb
  \field{journaltitle}{Mathematische Annalen}
  \field{month}{12}
  \field{year}{1928}
  \field{urlday}{05}
  \field{urlmonth}{03}
  \field{urlyear}{2015}
\endentry

\entry{tomamichel_strong_2014}{article}{}
  \name{author}{3}{}{%
    {{}%
     {Tomamichel}{T.}%
     {Marco}{M.}%
     {}{}%
     {}{}}%
    {{}%
     {Wilde}{W.}%
     {Mark~M.}{M.~M.}%
     {}{}%
     {}{}}%
    {{}%
     {Winter}{W.}%
     {Andreas}{A.}%
     {}{}%
     {}{}}%
  }
  \strng{namehash}{TMWMMWA1}
  \strng{fullhash}{TMWMMWA1}
  \field{abstract}{%
  We revisit a fundamental open problem in quantum information theory, namely
  whether it is possible to transmit quantum information at a rate exceeding
  the channel capacity if we allow for a non-vanishing probability of decoding
  error. Here we establish that the Rains information of any quantum channel is
  a strong converse rate for quantum communication: For any code with a rate
  exceeding the Rains information of the channel, we show that the fidelity
  vanishes exponentially fast as the number of channel uses increases. This
  remains true even if we consider codes that perform classical post-processing
  on the transmitted quantum data. Our result has several applications. Most
  importantly, for generalized dephasing channels we show that the Rains
  information is also achievable, and thereby establish the strong converse
  property for quantum communication over such channels. This for the first
  time conclusively settles the strong converse question for a class of quantum
  channels that have a non-trivial quantum capacity. Moreover, we show that the
  classical post-processing assisted quantum capacity of the quantum binary
  erasure channel satisfies the strong converse property.%
  }
  \field{title}{Strong converse rates for quantum communication}
  \verb{url}
  \verb http://arxiv.org/abs/1406.2946
  \endverb
  \verb{file}
  \verb Tomamichel et al_2014_Strong converse rates for quantum communication.p
  \verb df:/Users/joe/store/zotero/storage/C7DDQCXE/Tomamichel et al_2014_Stron
  \verb g converse rates for quantum communication.pdf:application/pdf
  \endverb
  \field{journaltitle}{arXiv:1406.2946 [quant-ph]}
  \field{month}{06}
  \field{year}{2014}
  \field{urlday}{07}
  \field{urlmonth}{02}
  \field{urlyear}{2015}
\endentry

\entry{devetak_private_2005}{article}{}
  \name{author}{1}{}{%
    {{}%
     {Devetak}{D.}%
     {Igor}{I.}%
     {}{}%
     {}{}}%
  }
  \keyw{channel capacity, channel coding, channel coding theorem, coherent
  information, cryptography, key generation, private classical information
  transmission, protocols, public classical communication, public key
  cryptography, pure bipartite entanglement, quantum channel capacity, quantum
  cryptography, quantum entanglement, wiretap channel}
  \strng{namehash}{DI1}
  \strng{fullhash}{DI1}
  \field{abstract}{%
  A formula for the capacity of a quantum channel for transmitting private
  classical information is derived. This is shown to be equal to the capacity
  of the channel for generating a secret key, and neither capacity is enhanced
  by forward public classical communication. Motivated by the work of
  Schumacher and Westmoreland on quantum privacy and quantum coherence,
  parallels between private classical information and quantum information are
  exploited to obtain an expression for the capacity of a quantum channel for
  generating pure bipartite entanglement. The latter implies a new proof of the
  quantum channel coding theorem and a simple proof of the converse. The
  coherent information plays a role in all of the above mentioned capacities.%
  }
  \verb{doi}
  \verb 10.1109/TIT.2004.839515
  \endverb
  \field{issn}{0018-9448}
  \field{number}{1}
  \field{pages}{44\bibrangedash  55}
  \field{title}{The private classical capacity and quantum capacity of a
  quantum channel}
  \field{volume}{51}
  \verb{file}
  \verb Devetak - 2005 - The private classical capacity and quantum capacit.pdf
  \verb :/Users/joe/store/zotero/storage/MM8TN2HU/Devetak - 2005 - The private
  \verb classical capacity and quantum capacit.pdf:application/pdf
  \endverb
  \field{journaltitle}{IEEE Transactions on Information Theory}
  \field{year}{2005}
\endentry

\lossort
\endlossort

  \blx@bblend
  \endgroup
  \csnumgdef{blx@labelnumber@\the\c@refsection}{0}}
\makeatother

\usepackage{braket}
\newcommand{\proj}[1]{|#1\rangle\langle #1|}
\def\tr{{\rm Tr}}

\def\eps{\varepsilon}
\renewcommand{\phi}{\varphi}
\def\id{\mathbbm{1}}
\def\jam{Jamio\l{}kowski }

\def\cB{\mathcal{B}}
\def\cC{\mathcal{C}}
\def\cD{\mathcal{D}}
\def\cE{\mathcal{E}}

\def\cG{\mathcal{G}}
\def\cH{\mathcal{H}}
\def\cI{\mathcal{I}}
\def\cK{\mathcal{K}}
\def\cM{\mathcal{M}}
\def\cN{\mathcal{N}}

\def\cR{\mathcal{R}}
\def\cS{\mathcal{S}}

\usepackage{xspace}

\newtheorem{theorem}{Theorem}
\newtheorem{lemma}{Lemma}

\newtheorem{proposition}{Proposition}
\newtheorem{corollary}{Corollary}

\usepackage{mathtools}
\def\defeq{\coloneqq}

\begin{document}

\title{A Minimax Converse for Quantum Channel Coding}
\author{Joseph M. Renes\\
{\small Institute for Theoretical Physics, ETH Z\"urich, 8093 Z\"urich, Switzerland}
}

\maketitle

\begin{abstract}
We prove a one-shot ``minimax'' converse bound for quantum channel coding assisted by positive partial transpose channels between sender and receiver.
The bound is similar in spirit to the converse by Polyanskiy, Poor, and Verd\'u [\href{http://dx.doi.org/10.1109/TIT.2010.2043769}{IEEE Trans.\ Info.\ Theory {\bf 56}, 2307–2359 (2010)}] for classical channel coding, and also enjoys the saddle point property enabling the order of optimizations to be interchanged. 
Equivalently, the bound can be formulated as a semidefinite program satisfying strong duality. 
The convex nature of the bound implies channel symmetries can substantially simplify the optimization, enabling us to explicitly compute the finite blocklength behavior for several simple qubit channels. 
In particular, we find that finite blocklength converse statements for the classical erasure channel apply to the assisted quantum erasure channel, while bounds for the classical binary symmetric channel apply to both the assisted dephasing and depolarizing channels.
This implies that these qubit channels inherit statements regarding the asymptotic limit of large blocklength, such as the strong converse or second-order converse rates, from their classical counterparts. Moreover, for the dephasing channel, the finite blocklength bounds are as tight as those for the classical binary symmetric channel, since coding for classical phase errors yields equivalently-performing unassisted quantum codes.  

This paper has been merged with~\cite{tomamichel_quantum_2015v1}; see~\cite{tomamichel_quantum_2015}.
\end{abstract}

\section{Introduction}
The capacity of a noisy channel is the ultimate, in-principle limit on its capability for reliable communication, and therefore studying channel capacity is an important goal in information theory. 
By its nature, the capacity is not of \emph{immediate} practical concern, as it ignores the resource requirements that would be needed to achieve the limit. 
Approaching capacity might, in principle, require coding operations and blocklengths too cumbersome or large to be implementable.
Nevertheless, several classical coding techniques developed in recent years have narrowed the gap between in-principle and in-practice for classical communication over classical channel, in particular polar codes~\cite{arikan_channel_2009} and spatially-coupled low-density parity-check codes~\cite{kudekar_spatially_2012}. 
Coding and decoding operations can be performed efficiently (quasilinearly) in the blocklengths of these codes, though the blocklengths themselves must still be rather large to approach capacity. 
The situation is dramatically different for quantum coding, where accurate control of quantum systems is a major experimental challenge, while manipulation and storage of classical bits is obscenely easy by comparison. 

Thus, it is of interest to better understand the possible performance of codes operating with fixed resources, in particular at finite blocklength. 
A bound which limits the performance of a coding scheme given fixed resources is known as a converse bound. 
For classical channels, the first truly systematic results on converse bounds limiting the size (blocklength) of codes with a given error probability were given by Polyanskiy, Poor, and Verd\'u~\cite{polyanskiy_channel_2010}. 
They formulated the converse bound in terms of a minimax optimization, and showed by numerical examples that the bound is quite tight for several channels of interest even at small blocklengths, by comparing to existing and novel achievability bounds. 
Subsequently, Matthews \cite{matthews_linear_2012} and Polyanskiy~\cite{polyanskiy_saddle_2013} demonstrated concavity and convexity properties of the bound which enable it to be formulated as a linear program (Matthews) or equivalently that the order of minimization and maximzation can be interchanged (Polyanskiy). 
Their results imply that channel symmetries can be used to simplify the optimization. 

For quantum channels, Matthews and Wehner extended the minimax approach to the task of transmitting classical information over quantum channels and formulated a bound in terms of a semidefinite program~\cite{matthews_finite_2014}. 
Recently, Leung and Matthews gave a semidefinite program for the optimal achievable fidelity for transmitting quantum information by codes of a fixed size~\cite{leung_power_2014}. 
In both of these cases the codes under consideration include the possibility of assistance by forward or reverse communication between sender and receiver, as this makes it possible to formulate the bounds as convex optimizations. 

In this paper we give a minimax bound for the size of codes for transmitting quantum information in terms of their entanglement fidelity. 
This is, in some sense, the opposite optimization as in~\cite{leung_power_2014}, and follows the original approach of~\cite{polyanskiy_channel_2010}. 
Here, too, assisting communication between sender and receiver is used to ensure the bound is tractable, and we show that the bound can be formulated as a semidefinite program. 
The advantage of this approach is that, as we show, simple qubit channels inherit converse bounds from simple classical channels, enabling us to directly apply results from~\cite{polyanskiy_channel_2010} and \cite{polyanskiy_saddle_2013} for classical channels to quantum problems. 

The remainder of the paper is organized as follows. 
The following section establishes the mathematical framework and notation used throughout. 
Then the precise details of the coding scenario under consideration and the minimax bound, Theorem~\ref{thm:minimaxbound}, are presented in Sec.~\ref{sec:converse}. Sec.~\ref{sec:sdp} shows that the minimax bound can be expressed as a semidefinite program. 
Next, Sec.~\ref{sec:symmetry} is a somewhat elaborate discussion of how symmetry can be used to simply the bound, which ultimately rests on its concavity and convexity properties. 
The particular qubit channel examples are detailed in Sec.~\ref{sec:examples}, and the paper finishes with a discussion of related bounds. 




\section{Mathematical Setup}

\subsection{States and Channels}

In this paper we consider finite-dimensional quantum systems, labelled by capital letters $A$, $B$, and so forth. 
The state space of system $A$ is denoted by $\cH_A$, and the dimension of this space $|A|$. 
The set of bounded operators on $\cH_A$ is denoted by $\cB(A)$, while the set of bounded operators with unit trace, i.e.\ the states of $A$, is denoted by $\cS(A)$.
The maximally mixed state on $A$ is denoted by $\pi_A=\id_A/|A|$. 
For two systems $A$ and $A'$ of the same dimension, i.e.\ with isomorphic state spaces, we may choose a basis $\{\ket{k}\}_k$ in each to define the unnormalized vector $\ket{\Omega}_{AA'}=\sum_{k}\ket{k}_{A}\otimes \ket k_{A'}$. 
The canonical maximally entangled state is then just $\ket{\Phi}_{AA'}=\frac 1{\sqrt{|A|}}\ket{\Omega}_{AA'}$.

A channel $\cN_{B|A}$ is a linear map from $\cB(A)$ to $\cB(B)$ which is both trace preserving and completely positive. 
Adjoints of channels and operators are denoted by $*$.
The Choi representative (or Choi operator) of the channel is the bipartite operator $N_{B|A}=(\cI_{A|A'}\otimes \cN_{B|A})[\Omega_{AA'}]$. 
The Choi representative of a channel satisfies $N_{B|A}\geq 0$ and $\tr_B[N_{B|A}]=\id_A$, and any bipartite operator $N_{B|A}$ satisfying these conditions defines a valid channel via $\rho_A\mapsto \tr_A[N_{B|A}\rho_A^T]$, a statement known as the Choi isomorphism~\cite{choi_completely_1975}. 
Note that trace non-increasing completely positive maps have Choi operators with $\tr_B[N_{B|A}]\leq \id_A$. 

In contrast, the \jam representative of $\cN_{B|A}$ is the operator $\hat N_{B|A}= N_{B|A}^{T_A}$, where $T_A$ denotes the transpose of system $A$.
Put differently, $\hat N_{B|A}=(\cI_{A|A'}\otimes \cN_{B|A})[\Omega_{AA'}^{T_A}]$~\cite{jamiolkowski_linear_1972}. 
An appealing property of the \jam representative is that the action of the channel no longer involves the transpose and is just $\rho_A\mapsto \tr_A[\hat N_{B|A}\rho_A]$. 
This formulation makes channel action look quite similar to marginalizing random variables in a classical probability distribution; for more on this point, see~\cite{leifer_towards_2013}. 
Further, the \jam isomorphism is natural, in the sense that it does not require a choice of basis, since $\Omega_{AA'}^{T_A}=\Pi_{AA'}^{\text{sym}}-\Pi_{AA'}^{\text{antisym}}$, and the projectors onto the symmetric and antisymmetric spaces of $\cH_A\otimes \cH_{A'}$ are basis-independent. 
However, we shall not make particular use of this fact here.  

Bipartite states $\rho_{AB}$ for which $\rho_{AB}^{T_B}\geq 0$ (equivalently, $\rho_{AB}^{T_A}\geq 0$) will be called positive partial transpose states, or PPT for short. 
A channel is PPT-preserving if a PPT input necessarily results in a PPT output. 
In~\cite{rains_semidefinite_2001}, it is shown that PPT-preserving channels have PPT Choi states (see the discussion after Eq.\ 4.13). 
More directly, suppose that $\phi_{AB}$ is a state with $\phi^{T_A}_{AB}\geq 0$ and $\cD_{MM'|AB}$ is a channel with $D_{MM'|AB}^{T_{MA}}\geq 0$. Let $\sigma_{MM'}=\cD_{MM'|AB}(\phi_{AB})$ and compute $\sigma_{MM'}^{T_M}$: 
\begin{align}
\sigma_{MM'}^{T_M} 
&=\tr_{AB}[D_{MM'|AB}^{T_M}\phi_{AB}^T]\\
&=\tr_{AB}[D_{MM'|AB}^{T_{MA}}\phi_{AB}^{T_A}].
\end{align}
Although the trace of positive operators is a positive number, the partial trace need not be a positive operator (consider $\tr_{B}[\Phi_{AB}\Phi_{BC}]$), so we cannot conclude that 
$\sigma_{MM'}^{T_M}\geq 0$ on this basis alone. However, the fact that the $\phi_{AB}^{T_A}$ factor is completely traced out, along with positivity of it and $D_{MM'|AB}^{T_{MA}}$, together imply that all expectation values of $\sigma_{MM'}^{T_M}$ are positive. 

\subsection{Semidefinite programming}
A semidefinite program (SDP) is simply an optimization of a linear function of a matrix or operator over a feasible set of inputs defined by positive semidefinite constraints. 
We give only the bare essentials here, for more detail see \cite{boyd_convex_2004,watrous_quantum_2011}.

The maximization form of an SDP is defined by a Hermiticity-preserving superoperator $\cE_{B|A}$ taking $\cB(A)$ to $\cB(B)$, a constraint operator $C\in \cB(B)$, and an operator $K\in \cB(A)$ which defines the objective function.  The SDP is the following optimization, which we will also refer to as the primal form,
\begin{align}
\begin{array}{r@{\,\,}rl}
\alpha=&\text{supremum} & \tr[KX]\\
&\text{subject to} & \cE(X)\leq C,\\
&& X\geq 0.
\end{array}
\end{align}
When the feasible set is empty, i.e.\ no $X$ satisfy the constraints, we set  $\alpha=-\infty$. 

The dual form arises as the optimal upper bound to the primal form, and takes the form 
\begin{align}
\begin{array}{r@{\,\,}ll}
\beta=&\text{infimum} & \tr[CY]\\
&\text{subject to} & \cE^*(Y)\geq K,\\
&& Y\geq 0.
\end{array}
\end{align}
Again, when the set of feasible $Y$ is empty, $\beta=\infty$. Weak duality is the statement that $\alpha\leq \beta$, that indeed the dual form gives upper bounds to the primal (or that the primal lower bounds the dual). 

Strong duality is the statement that the optimal upper bound equals the value of the primal problem, $\alpha=\beta$. This state of affairs often holds in problems of interest, and can be established by either of the following Slater conditions. In the first, called strict primal feasibility, strong duality holds if $\beta$ is finite and there exists an $X>0$ such that $\cE(X)<C$. Contrariwise, under strict dual feasibility strong duality holds when $\alpha$ is finite and there exists a $Y>0$ such that $\cE^*(Y)>K$. 
For strongly dual SDPs we also have the so-called complementary slackness conditions $\cE^*(Y)X=KX$ and $\cE(X)Y=CY$ that relate the primal and dual optimizers.

Neyman-Pearson hypothesis testing is an SDP particularly useful in the study of information processing protocols. Given one of two states $\rho$ or $\sigma$, called the null and alternate hypothesis, an asymmetric hypothesis test is a two-outcome POVM $\{\Gamma,\id-\Gamma\}$ that indicates which of the two hypotheses (states) is actually present. Here $\Gamma$ indicates $\rho$. Any such test makes two kinds of errors, the type-I error in which the null hypothesis is erroneously rejected, and the type-II error in which the alternate hypothesis is erroneously rejected. The probabilities of type-I and type-II errors are just $\tr[\rho(\id-\Gamma)]$ and $\tr[\sigma\Gamma]$, respectively. Fixing the probability of type-I error to $1-\eps$, we may ask for the test with the optimal (minimal) type-II error, called the Neyman-Pearson test. 
The optimal POVM is specified by $\Gamma$ with the smallest value of $\tr[\sigma\Gamma]$ such that $0\leq \Gamma\leq \id$ and $\tr[\rho\Gamma]\geq \eps$. 
This is a dual-form SDP with $Y=\Gamma$, $C=\sigma$, $K=(-\id,\eps)$ and $\cE^*(Y)=(-Y,\tr[\rho Y])$, and it satisfies strong duality. 
For future use, let us denote this optimal value by $\beta_\eps(\rho,\sigma)$:
\begin{align}
\begin{array}{r@{\,\,}ll}
\beta_\eps(\rho,\sigma)\defeq &\min & \tr[\Gamma \sigma]\\
&\text{subject to} & \tr[\Gamma\rho]\geq \eps\\
&& 0\leq \Gamma\leq \id.
\end{array}
\end{align}
The slackness conditions can be used to infer the form of the optimal test, recovering the  Neyman-Pearson lemma of classical statistics.

\section{Converse bound}
\label{sec:converse}

\subsection{Coding scenario}
\label{sec:codedef}
In this section we define the coding scenario precisely. 
Here we consider coding schemes for using a noisy quantum channel $\cN_{B|A}$ together with some auxiliary assistance channels to create a high-fidelity entangled state between sender Alice and receiver Bob. 
In particular, a PPT-assisted code denoted by $\cC=(\phi_{AA'B'},\cD_{MM'|BA'B'})$ consists of a PPT state $\phi_{AA'B'}$ with Alice in possession of systems $AA'$ and Bob $B'$ as well as a PPT-preserving decoding operation $\cD_{MM'|BA'B'}$ which outputs system $M$ to Alice and $M'\simeq M$ to Bob. By a slight abuse of notation, we also denote the dimension of system $M$ by $M$. Note that the decoding operation might consist of several rounds of forward and backward communication between sender and receiver. 

The coding scheme proceeds by Alice and Bob creating the state $\phi_{AA'B'}$ by means of a PPT channel, and Alice subsequently transmitting the $A$ system to Bob via $\cN_{B|A}$. Finally, they perform the operation $\cD_{MM'|BA'B'}$, leaving them with a bipartite state in $MM'$. 
The resulting entanglement fidelity is defined by  
\begin{align}
\label{eq:fidelitydef}
F(\cN;\cC)&\defeq\tr[\Phi_{MM'}\cD_{MM'|A'BB'}\circ\cN_{B|A}(\phi_{AA'B'})].
\end{align}
Finally, an $(M,\eps)$ code for $\cN_{B|A}$ is a pair $\cC=(\cE_{AA'B'},\cD_{MM'|BA'B'})$ with output dimension $M$ such that $F(\cN;\cC)\geq 1-\eps$.

As a side remark, when the initial state $\phi_{AA'B'}$ is separable and not only PPT, the coding scheme can be simplified without loss of generality. In particular, one can dispense with system $B'$ and restrict to pure states $\phi_{AA'}$. To see this, note that the fidelity is a convex combination of entanglement fidelities for the constitutent states $\sigma_{AA'}^j\otimes \eta_j^{B'}$ in $\phi_{AA'B'}=\sum_j \sigma_{AA'}^j\otimes \eta^j_{B'}$, so it will certainly not decrease when shifting to the code $\cC_j=(\sigma_{AA'}^j\otimes \eta^j_{B'},\cD_{MM'|BA'B'})$ with the largest entangelement fidelity. Further, as $\eta_B^j$ is fixed, it can be absorbed into the decoding operation. Finally, the same argument implies that one can further modify $\sigma_{AA'}^j$ to be a pure state. However, this argument does not necessarily go through in the general case when $\phi_{AA'B'}$ is PPT. 

\subsection{PPT-preserving channels}
Instead of the actual channel $\cN_{B|A}$, consider an arbitrary PPT-preserving channel $\cM_{B|A}$. This  produces a PPT state $\gamma_{MM'}=\cD_{MM'|A'BB'}\circ\cM_{B|A}(\phi_{AA'B'})$. The following argument, due to Rains, then immediately implies 
\begin{align}
\label{eq:pptfidelitybound}
F(\cM;\cC)\leq \frac1M.
\end{align} 
\begin{lemma}[Rains~\cite{rains_bound_1999}]
\label{lem:rainsppt}
For any subnormalized state $\gamma_{MM'}$ with $\gamma_{MM'}^{T_M}\geq 0$, $\tr[\Phi_{MM'}\gamma_{MM'}]\leq \frac 1M$. 
\end{lemma}
\begin{proof}
The proof proceeds by observing that the partial transpose of the maximally entangled state is the swap operator, normalized by $M$. 
\begin{align}
\tr[\Phi_{MM'}\gamma_{MM'}]&=\tr[\Phi_{MM'}^{T_M}\gamma_{MM'}^{T_M}]\\
&\leq ||\Phi_{MM'}^{T_M}||_\infty\tr[\gamma_{MM'}^{T_M}]\\
&\leq \frac1M.
\end{align}
The first inequality follows from the fact that  $\Phi_{MM'}^{T_M}=\frac1MU_{MM'}^{\textsc{swap}}$ and therefore $\Phi_{MM'}^{T_M}\leq ||\Phi_{MM'}^{T_M}||_\infty\id_{MM'}=\frac1M\id_{MM'}$. The second inequality is the subnormalization of $\gamma_{MM'}$. 
\end{proof}

\subsection{Fidelity in terms of the Choi operator}

Now we show that the entanglement fidelity \eqref{eq:fidelitydef} can be expressed directly in terms of its Choi operator of the channel and without explicit reference to the systems $MM'A'B'$. To do so, it is actually convenient to start with the \jam representation of channel action. 
In this representation we can write the fidelity as 
\begin{align}
F(\cN;\cC)&=\tr_{MM'}[\Phi_{MM'}\tr_{A'BB'}[\hat D_{MM'|A'BB'}\cN_{B|A}(\phi_{AA'B'})]]\\
&=\tr_{MM'}[\Phi_{MM'}\tr_{A'BB'}[\hat D_{MM'|A'BB'}\tr_{A}[\hat N_{B|A}\phi_{AA'B'}]]]\\
&=\tr_{MM'ABA'B'}[\Phi_{MM'}\hat D_{MM'|A'BB'}\hat N_{B|A}\phi_{AA'B'}]\\
&=\tr_{MM'ABA'B'}[\Phi_{MM'}\hat D_{MM'|A'BB'}\phi_{AA'B'}^{T_A}\hat N_{B|A}^{T_A}],
\end{align}
where, in the last equation, the order of $\phi$ and $N$ is interchanged in accordance with the usual rules of transposition. 

Define the operator
\begin{align}
\label{eq:codeLambda}
\Lambda_{AB}=\tr_{MM'A'B'}[\Phi_{MM'}\hat D_{MM'|A'BB'}\phi_{AA'B'}^{T_A}].
\end{align}
Since $N_{B|A}=\hat N_{B|A}^{T_A}$, the fidelity can be expressed as a linear function of the Choi operator of the channel,
\begin{align}
\label{eq:fidelitychoi}
F(\cN;\cC)=\tr[\Lambda_{AB} N_{B|A}],
\end{align}
for the particular $\Lambda_{AB}$ defined by the code. An $(M,\eps)$ code will have a $\Lambda_{AB}$ which satisfies 
\begin{align}
\label{eq:fidelitybound}
\tr[\Lambda_{AB} N_{B|A}]\geq 1-\eps.
\end{align}

Any such operator $\Lambda_{AB}$ satisfies the following two simple properties 
\begin{proposition}
\label{prop:Lambda}
Any operator $\Lambda_{AB}$ defined as in \eqref{eq:codeLambda} satisfies
\begin{align}
&0\leq \Lambda_{AB}\leq \phi_A^T\otimes \id_B,\label{eq:lambdabounds1}
\end{align}
\end{proposition}
\begin{proof}
First regard $\phi^{T_A}_{AA'B'}$ as the \jam representative $\hat R_{A'B'|A}=\phi_{AA'B'}^{T_A}$ of a channel $\cR_{A'B'|A}$. Then 
\begin{align}
\Lambda_{AB}=\cR_{A'B'|A}^*\circ\cD_{MM'|BA'B'}^*(\Phi_{MM'}).
\end{align}
Now observe that $\cR_{A'B'|A}$ is completely positive and trace descreasing:
\begin{align}
&\hat R_{A'B'|A}^{T_A}=\phi_{AA'B'}\geq 0,\\
&\tr_{A'B'}[\hat R_{A'B'|A}]=\phi_A^T\leq \id_A,
\end{align}
meaning its adjoint action is 
completely positive and subunital. 
Therefore, \eqref{eq:lambdabounds1} follows from the fact that $0\leq \Phi_{MM'}\leq \id_{MM'}$. 

\end{proof}
These constraints hold for more than just PPT-assisted codes, as is easily demonstrated. Let $\cD_{MM'|BA'B'}$ be the channel which creates $\Phi_{MM'}$ and ignores (traces out) the input systems $BA'B'$. Then $\hat D_{MM'|BA'B'}=\Phi_{MM'}\id_{BA'B'}$. This leads to $\Lambda_{AB}=\phi_A^T\id_B$, which satisfies both constraints.

\subsection{Minimax converse bound}
It is now straightforward to derive the minimax converse bound. 
The idea is simple: Maximizing over PPT channels $\cM_{B|A}$ in the fidelity expression \eqref{eq:fidelitychoi} and using \eqref{eq:pptfidelitybound} gives a lower bound on $\frac 1M$, i.e.\ an upper bound on $M$. 
However, the result depends on the details of the code via $\Lambda_{AB}$. 
This dependence can be removed by minimizing $\Lambda_{AB}$ over the smallest conveniently-described set which certainly contains the $\Lambda_{AB}$ associated with the code. 
Here, this is the set defined by the constraints in \eqref{eq:fidelitybound} and \eqref{eq:lambdabounds1}, for arbitrary subnormalized $\phi_A$, as these do not depend on the precise details of the coding operations. 

To state the minimax bound formally, first define the following sets 
\begin{align}
\textsc{ppt}&\defeq\{M_{B|A}: M_{B|A}\geq 0, M_{B|A}^{T_A}\geq 0, \tr_B[ M_{B|A}]\leq \id_A\}\label{eq:pptset}\\
\textsc{f}(\cN,\eps)&\defeq\{(\phi_A,\Lambda_{AB}):\phi_A\geq 0,\tr[\phi_A]\leq 1, 0\leq \Lambda_{AB}\leq \phi_A^T\id_B,\tr[\Lambda_{AB}N_{B|A}]\geq 1-\eps\}.\label{eq:feasibleset}
\end{align}
Then we have 
\begin{theorem}
\label{thm:minimaxbound}
Any $(M,\eps)$ PPT-assisted code satisfies
\begin{align}
\label{eq:minimaxbound}
\min_{(\phi_A,\Lambda_{AB})\in \textsc{f}(\cN,\eps)} \max_{M_{B|A}\in \textsc{ppt}}\tr[\Lambda_{AB} M_{B|A}]\leq \frac1M.
\end{align}
\end{theorem}
For later convenience, let us define 
\begin{align}
f(\cN,\eps)\defeq\min_{(\phi_A,\Lambda_{AB})\in \textsc{f}(\cN,\eps)} \max_{M_{B|A}\in \textsc{ppt}}\tr[\Lambda_{AB} M_{B|A}].
\end{align}
Before proceeding with the proof, observe that we can interchange the order of optimization in $f$, due to von Neumann's minimax theorem~\cite{neumann_zur_1928}, as the objective function is linear and both $\textsc{f}(\cN,\eps)$ and $\textsc{ppt}$ are compact, convex sets. This gives 
\begin{corollary}
Any $(M,\eps)$ PPT-assisted code satisfies
\begin{align}
\label{eq:maximinbound}
 \max_{M_{B|A}\in \textsc{ppt}}\min_{(\phi_A,\Lambda_{AB})\in \textsc{f}(\cN,\eps)}\tr[\Lambda_{AB} M_{B|A}]\leq \frac1M.
\end{align}
\end{corollary}

\begin{proof}[Proof of Theorem~\ref{thm:minimaxbound}]
Consider then the following function,
\begin{align}
\label{eq:f0def}
f_0(O_{AB})\defeq \sup_{M_{B|A}\in \textsc{ppt}}\tr[O_{AB} M_{B|A}],
\end{align}
defined on all bipartite operators $\{O_{AB}:0\leq O_{AB}\leq \id_{AB}\}$. 
First note that the supremum is attained, since the objective function is a continuous function $M\mapsto \tr[OM]$ on a compact set, $\textsc{ppt}$. 
Moreover, $f_0$ is continuous; specifically, it obeys
\begin{align}
|f_0(O'_{AB})-f_0(O_{AB})|\leq |A|\,||O'_{AB}-O_{AB}||_1.
\end{align}
To see this, suppose that $f_0(O_{AB})\leq f_0(O'_{AB})$, otherwise swap the two. 
Then $|f_0(O'_{AB})-f_0(O_{AB})|=f_0(O'_{AB})-f_0(O_{AB})$. 
Let $M_{B|A}'$ be the optimizer in $f_0(O_{AB}')$. 
By the variational characterization of the trace norm, 
\begin{align}
f_0(O_{AB})&\geq \tr[M_{B|A}'O_{AB}]\\
&\geq \tr[M_{B|A}'O'_{AB}]-|A|\,||O_{AB}'-O_{AB}||_1\\
&=f_0(O_{AB}')-|A|\,||O_{AB}'-O_{AB}||_1.
\end{align}

For $\Lambda_{AB}$ defined from an $(M,\eps)$ code as in \eqref{eq:codeLambda}, \eqref{eq:pptfidelitybound} implies $f_0(\Lambda_{AB})\leq \frac 1M$. Taking the infimum over $(\phi_A,\Lambda_{AB})\in \textsc{f}(\cN,\eps)$ gives a bound independent of the precise details of the code. Finally, again since $\textsc{f}(\cN,\eps)$ is convex and compact and $f_0$ is continuous, the infimum is attained. 
\end{proof}

\section{The minimax bound as a semidefinite program}
\label{sec:sdp}
In this section we describe how to formulate the minimax bound as a semidefinite program satisfying strong duality. Doing so is straightforward: We simply use the dual of the inner optimization in \eqref{eq:minimaxbound} to obtain a minimization problem, or the dual the of the inner optimization in \eqref{eq:maximinbound} to obtain a maximization problem. Ultimately we find the following
\begin{proposition}
For any channel $\cN_{B|A}$ and $0\leq \eps\leq 1$,
\begin{align}
\label{eq:sdps}
\begin{array}{crlcrl}
 f(\cN,\eps)\,\,= &{\text{minimum}} &\tr[\xi_A]&\!\!\!=&{\text{maximum}} &m(1-\eps)-n\\
& \text{subject to} & \xi_A\id_B\geq \Lambda_{AB}+\Gamma_{AB}^{T_A}, &&\text{subject to} & M_{B|A}\in \textsc{ppt},\\
&& (\phi_A,\Lambda_{AB})\in \textsc{f}(\cN,\eps),&&&mN_{B|A}
\leq  M_{B|A}+R_{AB}
,\\
&& \phi_A,\Lambda_{AB},\Gamma_{AB},\xi_A\geq 0&&&\tr_B[R_{AB}]
\leq n \id_A,\\
&&&&&m,n,R_{AB}\geq 0
\end{array}
\end{align}
\end{proposition}

\begin{proof}
Let us take the former approach, dualize the inner optimization in \eqref{eq:maximinbound}, and then show that strong duality holds. Of course, strong duality must hold, as it is equivalent to the saddle point property, but we shall give a simple independent argument for strong duality based on Slater's condition. 

Observe that $f_0$ is a semidefinite program, in particular, a primal problem as we have defined it, with $X=M_{B|A}$, $K=O_{AB}$, $C=(0,\id_A)$, and $\cE(X)=(-X^{T_A},\tr_B[X])$.
Choosing for the dual variables $Y=(\Gamma_{AB},\xi_A)$, the dual of $f_0$ is 
\begin{align}
\begin{array}{crl}
\tilde f_0(O_{AB})\defeq &\underset{\Gamma_{AB},\xi_A}{\text{minimum}} &\tr[\xi_A]\\
& \text{subject to} & \xi_A\id_B\geq O_{AB}+\Gamma_{AB}^{T_A}, \\
&& \Gamma_{AB},\xi_A\geq 0.
\end{array}
\end{align}
Combining this with the outer optimization over $\textsc{f}(\cN,\eps)$ gives the minimization program in \eqref{eq:sdps}. 

The equality statement is precisely strong duality of the primal and dual forms of the inner optimization. By Slater's condition, strong duality holds if $f_0$ is finite and there exists a strictly feasible set of dual variables. 
Observe that $f_0(O_{AB})\leq |A|$, since for the optimal $M_{B|A}$ we have $f_0(O_{AB})=\tr[M_{B|A}O_{AB}]\leq \tr[M_{B|A}]\leq \tr_A \id A=|A|$. Here we have used the upper bounds $O_{AB}\leq \id_{AB}$ and $\tr_B[M_{B|A}]\leq \id_A$. Thus, the first condition is fulfilled. Meanwhile, $\Gamma_{AB}=\id_{AB}$ and $\xi_A=3\id_A$ are a strictly feasible pair. 
Thus, $\tilde f_0=f_0$ over the domain of interest.

To construct the maximization program, we simply dualize the minimzation program. 
In particular, $f(\cN,\eps)$ is a dual-form semidefinite program in the variable $Y=(\phi_A,\Lambda_{AB},\Gamma_{AB},\xi_A)$ with $C=(0,0,0,\id_A)$, $K=(1-\eps,-1,0,0)$, and 
\begin{align}
\cE^*(Y)=(\tr[N_{B|A}\Lambda_{AB}],-\tr[\phi_A],\phi_A^T\id_B-\Lambda_{AB},\xi_A\id_B-\Lambda_{AB}-\Gamma^{T_A}_{AB}).
\end{align}
Choosing primal variables $X=(m,n,R_{AB},M_{AB})$ leads to the maximization in \eqref{eq:sdps}. 

Equality again follows from Slater's condition: $f$ is finite by the minimax formulation (in particular the bound on $f_0$ used above), while 
a feasible choice of dual variables is given by $M_{AB}=R_{AB}=\frac 1{2|B|}\id_{AB}$, $n=1$, and $m=\frac 1{2|A||B|}$. 
The choice of $m$ ensures the first constraint holds strictly, since any Choi operator of a trace-preserving map satisfies $|| N_{B|A}||_\infty=|A|$. 
\end{proof}

No discussion of strong duality of semidefinite programs is complete until the complementary slackness conditions have been formulated. 
Often, these give considerable insight into the form and properties of the optimizing variables. First observe that 
\begin{align}
\cE(X)=(-n\id_A+\tr_B[R_{AB}^{T_A}]
,\,mN_{B|A}-M_{B|A}-R_{AB}
,\,-M_{AB}^{T_A},\,\tr_B[M_{B|A}]).
\end{align}
Then the conditions are easy to read off from the form of $C$ and $K$. They are 
\begin{align}
\tr[\phi_A]&=1\\
\tr[\Lambda_{AB}N_{B|A}]&=1-\eps,\\
\phi_A^T R_{AB}&=\Lambda_{AB}R_{AB}\\
\xi_AM_{B|A}&=(\Lambda_{AB}+\Gamma^{T_A}_{AB})M_{B|A}\\
n\phi_A&=\tr_B[R_{AB}^{T_A}
]\phi_A\\
M_{B|A}^{T_A}\Gamma_{AB}&=0\\
\tr_B[M_{B|A}]\xi_A&=\xi_A,\\
m N_{B|A}\Lambda_{AB}&=(M_{B|A}+R_{AB}
)\Lambda_{AB}.
\end{align}

\section{Channel symmetry}
\label{sec:symmetry}
Symmetries of the channel can greatly simplify the calculation of the minimax bound. 
First let us state precisely what we mean by channel symmetries.  
Suppose $G$ is a group, possibly a topological group, represented by operators $U_g$ on $A$ and $V_g$ on $B$. 
A channel $\cN_{B|A}$ is covariant with respect to $G$ when  
\begin{align}
V_g \cN(\cdot)V_g^* = \cN(U_g \cdot U_g^*)\qquad \forall g\in G.
\end{align}
We can write this as an invariance of the channel:
\begin{align}
\cN(\cdot)= V_g^*\cN(U_g \cdot U_g^*)V_g\qquad \forall g\in G.
\end{align}
In terms of the Choi operator, the condition is simply
\begin{align}
(U_g^T)_A\otimes (V_g^*)_B N_{B|A}(\overline{U}_g)_A\otimes (V_g)_B= N_{B|A}\qquad \forall g \in G.
\end{align}
Thus, the Choi state is a fixed point when averaging over the action of the group. 
To enforce such averaging, introduce the superoperator $\cG_{AB}$: 
\begin{align}
\cG(O_{AB})=\int_{G}\!\!\!\text{d}\mu(g)\,\,(U_g^T)_A\otimes (V_g^*)_B O_{AB}(U_g^T)^*_A\otimes (V_g)_B,
\end{align}
where  $\mu$ is the Haar measure of the group. 
Observe that $\cG^*=\cG$, since taking the adjoint of the group elements just reparameterizes the group, sending $g$ to $g^{-1}$. 

Due to the structure of the symmetrization $\cG$, we have the following 
 \begin{proposition}
 \label{prop:symmetrizefeasible}
Suppose $(\phi_A,\Lambda_{AB})\in \textsc{f}(\cN,\eps)$. Then $(\cG(\phi_A),\cG(\Lambda_{AB}))\in \textsc{f}(\cN,\eps)$. 
Similarly, $\cG(M_{B|A})\in \textsc{ppt}$ for any $M_{B|A}\in \textsc{ppt}$.
\end{proposition}
\begin{proof}
Start with the latter claim, and let $\bar M_{B|A}=\cG(M_{B|A}). $The positivity condition in \eqref{eq:pptset} holds for $\bar M_{B|A}$ since $\cG$ is completely positive. For the trace condition, we have
\begin{align}
\tr_B[\bar M_{B|A}]&=\int_{G}\!\!\!\text{d}\mu(g)\,\,(U_g^T)_A \tr_B [(V_g^*)_B M_{B|A} (V_g)_B
] (U_g^T)^*_A\\
&\leq \int_{G}\!\!\!\text{d}\mu(g)\,\,(U_g^T)_A \id_A (U_g^T)^*_A\\
&=\id_A
\end{align}
For the partial transpose condition, note that for any operator $O_{AB}$, 
\begin{align}
\label{eq:aveOT}
\bar O_{AB}^{T_A}=\int_{G}\!\!\!\text{d}\mu(g)\,\,(U_g)_A\otimes (V_g^*)_B O_{AB}^{T_A}(U_g)^*_A\otimes (V_g)_B,
\end{align}
which is the action of a slightly different, yet still-completely-positive, version of $\cG$ on $O_{AB}$. Thus, $\bar M_{B|A}^{T_A}$ is positive if $M_{B|A}^{T_A}$ is. 

For the former claim, again let $\bar\phi_A=\cG_{A}(\phi_A)$ and $\bar\Lambda_{AB}=\cG_{AB}(\Lambda_{AB})$. Returning to \eqref{eq:feasibleset}, the positivity conditions $\bar \phi_A,\bar \Lambda_{AB}\geq 0$ hold because $\cG$ is completely positive.
The two trace conditions still hold because $\cN_{B|A}$ is $G$-covariant.  
It remains to show the upper bound on $\bar\Lambda_{AB}$:
\begin{align}
\bar\Lambda_{AB}&\leq \int_{G}\!\!\!\text{d}\mu(g)\,\,(U_g^T)_A\otimes (V_g^*)_B (\phi_A^T\otimes \id_B) (U_g^T)^*_A\otimes (V_g)_B \\
&=\int_{G}\!\!\!\text{d}\mu(g)\,\,(U_g^T\phi^T \overline{U}_g)_A\otimes \id_B\\
&=\big(\int_{G}\!\!\!\text{d}\mu(g)\,\,U_g^* \phi {U}_g\big)_A^T\otimes \id_B\\
&=\phi_A^T\id_B.
\end{align}
\end{proof}

As shown by Polyanskiy for the classical metaconverse~\cite{polyanskiy_saddle_2013}, we can now show that $G$-covariant quantum channels have $G$-invariant optimizers. Letting $\textsc{ppt}^G$ be the set of Choi operators of PPT-preserving channels which are invariant under $\cG$, and similarly $\textsc{f}^G(\cN,\eps)$ the intersection of $\textsc{f}(\cN,\eps)$ with $G$-invariant operators, we have
\begin{theorem}
For $G$-covariant channels $\cN_{B|A}$, we can restrict the optimizations in \eqref{eq:minimaxbound} to $\textsc{ppt}^G$ and $\textsc{f}^G(\cN_{B|A},\eps)$.
\end{theorem}
\begin{proof}
The proof proceeds similarly to that of \cite[Theorem 20]{polyanskiy_saddle_2013}. 
To simplify notation, define 
\begin{align}
g(O_{AB})=(U_g^T)_A\otimes (V_g^*)_B O_{AB}(U_g^T)^*_A\otimes (V_g)_B.
\end{align} 
Now consider the outer optimization in \eqref{eq:minimaxbound}. First note that the function $f_0$ from \eqref{eq:f0def} is convex, as it is the pointwise maximum of linear functions.
Furthermore, $f_0$ must be constant on orbits of $G$. Suppose $M_{B|A}'$ is the optimizer for $g(O_{AB})$ for some arbitrary $g\in G$, so that $f_o(O_{AB})=\tr[M_{B|A}' g(O_{AB})]$. It follows that $g^{-1}(M'_{B|A})$ is feasible for $f_0(O_{AB})$, since independent unitary operations on $A$ and $B$ are PPT-preserving. Hence, $f_0(O_{AB})\geq f_0(g(O_{AB}))$. But the same argument implies $f_0(g(O_{AB}))\geq f_0(g^{-1}\circ g(O_{AB}))=f_0(O_{AB})$. 

Applying Jensen's inequality and taking the minimum over  $(\phi_A,\Lambda_{AB})\in \textsc{f}(\cN,\eps)$ gives
\begin{align}
\min_{(\phi_A,\Lambda_{AB})\in \textsc{f}(\cN,\eps)}f_0(\cG(\Lambda_{AB}))\leq \min_{(\phi_A,\Lambda_{AB})\in \textsc{f}(\cN,\eps)} f_0(\Lambda_{AB}).
\end{align}
By Proposition~\ref{prop:symmetrizefeasible}, we can restrict the optimization on the lefthand side to $F^G(\cN,\eps)$ and obtain
\begin{align}
\min_{(\phi_A,\Lambda_{AB})\in \textsc{f}^G(\cN,\eps)}f_0(\Lambda_{AB})\leq \min_{(\phi_A,\Lambda_{AB})\in \textsc{f}(\cN,\eps)} f_0(\Lambda_{AB}).
\end{align}
On the other hand, since we are now minimizing over a smaller set, the lefthand side of this expression cannot be smaller than the right, so equality holds.  

Next consider the inner optimization in \eqref{eq:minimaxbound}, with the additional restriction to $\textsc{f}^G(\cN,\eps)$ in the outer optimization. For $f_0(\Lambda_{AB})$ with $G$-invariant $\Lambda_{AB}$, the objective function does not change upon replacing $M_{B|A}$ by $\cG(M_{B|A})$. Proposition~\ref{prop:symmetrizefeasible} then implies that we can safely restrict the optimization to $\textsc{ppt}^G$. 
\end{proof}

\section{Examples}
\label{sec:examples}
\subsection{Qubit dephasing channel}
Here we show that both finite blocklength converse and achievability bounds for the qubit dephasing channel can be inherited from the corresponding bounds in~\cite{polyanskiy_channel_2010} for the classical binary symmetric channel (BSC). 
The dephasing channel $\cN_{B|A}$ has Kraus operators $\sqrt{1-p}\id$ and $\sqrt{p}\sigma_z$, and the Choi operator is $N_{B|A}=2(1-p)\Phi^+_{AB}+2p\Phi^-_{AB}$. Here $\Phi^+$ is the canonical maximally entangled state, i.e.\ $\Phi^+=\Phi$, and $\Phi^-=(\id\otimes \sigma_z)\Phi(\id\otimes \sigma_z)$. 
Since the Bell states are orthogonal, dephasing is essentially the BSC for phase errors.

Clearly $\cN_{B|A}$ is covariant under the action of $\sigma_z$; indeed, it is covariant under any unitary operator diagonal in the dephasing basis.  Since $\sigma_x \sigma_z\sigma_x=-\sigma_z$, it is covariant under the action of $\sigma_x$ as well. For a single qubit the relevant symmetry group is $G_1=\{\id,\sigma_x,\sigma_y,\sigma_z\}$ (up to phases, which are irrelevant since the group acts by conjugation by the Pauli operators). The corresponding $\cG_1$ on the input space has the action $\cG_1(\rho)=\pi$ for all $\rho$, where $\pi$ is the maximally mixed state. 

The memoryless extension $\cN^{(n)}_{B^n|A^n}=\cN_{B|A}^{\otimes n}$ inherits these symmetries in each input space, so $G_n=G_1\times G_1\times \cdots\times G_1$. 
Similarly, $\cG_n(\rho_n)=\pi_n$ has the effect of completely depolarizing the input state.  That is, the optimal input state $\phi_{A^n}$ for the bound is maximally mixed.
Thus, 
\begin{align}
\textsc{f}^G(\cN^{\otimes n}_{B|A},\eps)=\{\Lambda_{A^nB^n}:0\leq \Lambda_{A^nB^n}\leq \tfrac1{2^n}\id_{A^nB^n},\tr[\Lambda_{A^nB^n} N_{B|A}^{\otimes n}]\geq 1-\eps\}.
\end{align}
 Furthermore, we can restrict $\Lambda_{A^nB^n}$ to be in the support of $N_{B|A}^{\otimes n}$ without loss of generality, since feasibility will not be affected and the objective function in the inner minimization of the maximin bound \eqref{eq:maximinbound} can only decrease. 
 Staying with the maximin bound, we may choose $\cM_{B^n|A^n}$ to have Choi state ${M}_{B|A}^{\otimes n}$ with $M_{B|A}=\Phi^+_{AB}+\Phi^-_{AB}$, i.e.\ the fully dephasing channel. Defining $L_{A^nB^n}=2^n \Lambda_{A^nB^n}$, \eqref{eq:maximinbound} then yields 
\begin{align}
\begin{array}{cccrl}
\frac1M &\geq &\underset{L_{A^nB^n}}{\text{minimize}} &\tr[L_{A^nB^n} \cM_{B|A}(\Phi_{AB})^{\otimes n}]\\
&& \text{subject to} & \tr[L_{A^nB^n} \cN_{N|A}(\Phi_{AB})^{\otimes n}]\geq 1-\eps, \\
&&& 0\leq L_{A^nB^n}\leq \id_{A^nB^n}.
\end{array}
\end{align}
Letting $\omega_{AB}=\cN_{B|A}(\Phi_{AB})=(1-p)\Phi^+_{AB}+p\Phi^-_{AB}$ and $\sigma_{AB}=\cM_{B|A}(\Phi_{AB})=\tfrac12(\Phi^+_{AB}+\Phi^-_{AB})$, the righthand side is just the minimal type-II error of distinguishing $\omega_{AB}^{\otimes n}$ from $\sigma_{AB}$, for type-I error constrained to be no larger than $\eps$. That is, 
\begin{align}
\label{eq:dephasingbound}
\frac1M\geq \beta_{1-\eps}(\omega_{AB}^{\otimes n},\sigma_{AB}^{\otimes n}).
\end{align}

Since both states are diagonal in the same basis, the hypothesis test between $\omega_{AB}$ and $\sigma_{AB}$ can be recast as a test between classical distributions. 
Observe that measuring each system of $\omega_{AB}$ and $\sigma_{AB}$ in the basis of $\sigma_x$ produces the probability distributions $P_{XY}$ and $P_XQ_Y$, respectively, with $P_X$ and $Q_Y$ uniformly distributed and $P[X=Y]=1-p$.  
Moreover, we can reconstruct the original states $\omega_{AB}$ and $\sigma_{AB}$ from $P_{XY}$ and $P_XQ_Y$ with the map that sends $(X,Y)$ to $\Phi^+_{AB}$ when $X=Y$ and to $\Phi^-_{AB}$ otherwise. Therefore we have
\begin{align}
\label{eq:dephasingbsc}
  \frac1M\geq \beta_{1-\eps}(P_{XY}^{\times n},P_X^{\times n}Q_Y^{\times n}).
\end{align}
This bound is precisely the expression obtained for the binary symmetric channel by Polyanskiy, Poor, and Verd\'u~\cite[Theorem 26]{polyanskiy_channel_2010} (see also~\cite[Theorem 22]{polyanskiy_saddle_2013}). 
Thus, the dephasing channel inherits the finite blocklength converse of the BSC, and the bounds are identical if our choice for $\cM_{B^n|A^n}$ is optimal. 

Regardless of the optimality of $\cM_{B^n|A^n}$, asymptotic results such as the strong converse and second order coding rate for the dephasing channel follow directly from the classical problem, in particular Eq.\ 160 and Theorem 52 in~\cite{polyanskiy_channel_2010}, respectively.
Alternately, one can deduce the strong converse property more immediately by simply invoking Stein's lemma on \eqref{eq:dephasingbound} or \eqref{eq:dephasingbsc}. (Note that the strong converse was first shown in~\cite{tomamichel_strong_2014}.)

Finally, the achievable bounds for dephasing are also at least as good as those of the classical BSC, simply because any classical code for the BSC can be regarded as correcting phase errors and applied to the dephasing channel. More specifically, the classical code can be used as part of a CSS-like quantum code, as described in~\cite{devetak_private_2005}. The ``error correction'' part of the code (see the Remark prior to \S V) is just the code for the BSC, applied to the $\sigma_x$ basis (i.e.\ with inputs diagonal in this basis). We can dispense with the ``privacy amplification'' part of the code, since it may be easily verified that the complement of the dephasing channel has a constant output on inputs diagonal in the $\sigma_x$ basis. We require a CSS-like code and not a proper CSS code because the classical BSC code need not be a linear code. Note also that, importantly, the guessing probability of the classical code is equal to the fidelity of the quantum code for this channel. For a more detailed and convincing discussion of this point, see~\cite{tomamichel_quantum_2015}.

\subsection{Erasure channel}
For the qubit erasure channel we can inherit a converse bound from the metaconverse of the classical binary erasure channel (BEC). 
The qubit erasure channel has qubit input and output $B$ of dimension three, namely $B=A\oplus \mathbb C$. The extra dimension indicates to the receiver that the input was erased.   
The Choi state of the erasure channel with probability $p$ is simply $N_{B|A}=2(1-p)\Phi_{AB}+2p \pi_A\otimes \proj e_B$, where $\ket{e}$ is the additional vector in $B$. 
The channel is covariant with respect to action by any unitary on the input, with corresponding inverse on the output, plus dephasing of the output into the $A$ and $\ket{e}$ subspaces. 
Therefore, the optimal input state is the maximally mixed state. Let $\omega_{AB}^{\otimes n}$ be the output of the $n$-fold application of $\cN_{B|A}$ and let $\sigma_{A^nB^n}$ be the output of the PPT map $\cM_{B^n|A^n}$. 

As with the dephasing channel, consider a measurement of $A^n$ and $B^n$ in the standard basis and call the output random variables $X^n$ and $Y^n$, respectively. 
From $\omega_{AB}^{\otimes n}$ we obtain the distribution $P_{XY}^{\otimes n}$, with $P_{X}$ uniform and $Y=X$ with probability $1-p$ and equal to $e$ with probability $p$. 
Note that we can recover $\omega_{AB}$ from $P_{XY}$ by employing the map which produces $\pi_A\otimes \proj e_B$ when $Y=e$ and otherwise $\Phi^+_{AB}$.
Now let $\sigma_{A^nB^n}$ be the state obtained by this map for the distribution $P_{X^n}Q_{Y^n}$ with $Q_{Y^n}$ the optimal choice for the metaconverse of the classical BEC, Eq.~168 of~\cite{polyanskiy_saddle_2013}. 
Due to the product form of the classical distribution, $\sigma_{A^nB^n}$ can be obtained from a PPT channel acting on the maximally entangled state. Thus, we have
$\beta_{1-\eps}(\omega_{AB}^{\otimes n},\sigma_{A^nB^n})=\beta_{1-\eps}(P_{XY}^{\otimes n},P_{X}^{\otimes n}Q_{Y^n})$, where the latter quantity appears in the converse bound for the classical BEC given in~\cite{polyanskiy_saddle_2013}.
Therefore, by the minimax bound we obtain
\begin{align}
\frac1M\geq \beta_{1-\eps}(P_{XY}^{\otimes n},P_{X}^{\otimes n}Q_{Y^n}),
\end{align}
meaning the minimax bound for the BEC also applies to the qubit erasure channel with PPT assistance. Again, we may infer asymptotic statements such as the strong converse and second order coding rates from this bound. 



\subsection{Depolarization}
The depolarizing channel has Choi state $N_{B|A}=2(1-p)\Phi^+_{AB}+\frac{2p}{3}(\Phi^-_{AB}+\Psi^+_{AB}+\Psi^-_{AB})$, where $p$ is the probability of depolarization, and $\Psi^\pm$ are obtained from $\Phi^+$ by conjugation with $\sigma_x$ and $\sigma_y$, respectively.
The symmetry group of this channel includes all unitary operations, meaning that the optimal $\phi_A$ is again the maximally mixed state. If we choose $M_{B^n|A^n}=M_{B|A}^{\otimes n}$ with $M_{B|A}=\Phi^+_{AB}+\frac 13(\Phi^-_{AB}+\Psi^+_{AB}+\Psi^-_{AB})$, the minimax bound involves the optimal hypothesis test between $n$ copies of $\omega_{AB}=(1-p)\Phi^+_{AB}+\frac p 3(\Phi^-_{AB}+\Psi^+_{AB}+\Psi^-_{AB})$ and $n$ copies of $\sigma_{AB}=\frac12\Phi^+_{AB}+\frac16(\Phi^-_{AB}+\Psi^+_{AB}+\Psi^-_{AB})$:
As in the case of dephasing, we can convert the hypothesis test between $\omega_{AB}$ and $\sigma_{AB}$ into a test between classical distributions, in fact precisely those distributions which were used in the dephasing example. This follows by considering the map which generates $\Phi^+$ when $X=Y$ and otherwise randomly generates one of the other Bell states when $X\neq Y$. Therefore, we obtain the same bound, \eqref{eq:dephasingbsc}, for depolarization as for dephasing.

This raises the question of whether PPT assistance can turn the depolarizing channel into the dephasing channel. To investigate this further, a sensible first step would be to establish optimality of the two bounds, to ensure they are truly equivalent. 

%

\section{Discussion}

We have derived a minimax bound for the size of a PPT-assisted quantum code given a target entanglement fidelity, very much along the lines of the classical bound by Polyanskiy, Poor, and Verd\'u~\cite{polyanskiy_channel_2010}. 
The restriction to PPT-assistance comes from the use of Rains's bound, Lemma~\ref{lem:rainsppt}, though in principle, the bound applies to all channels $\cM_{B|A}$ which deliver a state having overlap $1/M$ with the maximally entangled state. Focussing on PPT has the advantage that the PPT conditions can be phrased as linear constraints, and lead to a semidefinite program formulation of the bound.  

It would be desirable to incorporate the PPT constraint into the feasible set $\textsc{f}(\cN,\eps)$ itself, so as to tighten the bound. Additional constraints can indeed be found along the lines of Leung and Matthews~\cite{leung_power_2014}. 
However, it appears to be impossible to incorporate such additional constraints on $\textsc{f}(\cN,\eps)$ and obtain a bound, like that of Theorem~\ref{thm:minimaxbound}, in which we optimize the code size for fixed target fidelity. 
The difficulty is that the further constraints on $\textsc{f}(\cN,\eps)$ directly involve $M$, the size of the code. Leung and Matthews avoid this problem by optimizing the target fidelity for fixed code size, rather than the other way around.

The analog of this issue in the classical case is that the metaconverse in \cite{polyanskiy_channel_2010} also applies to non-signalling assisted codes~\cite{matthews_linear_2012,polyanskiy_saddle_2013}, not just unassisted codes. There, however, no great gain is to be had by adding further constraints to the analog of $\textsc{f}(\cN,\eps)$, potentially tightening the bound: The $1/M$ bound is already quite tight at moderate blocklengths for channels of interest, as shown in~\cite{polyanskiy_channel_2010}. 



One might also hope to obtain a useful bound for unassisted codes by subsituting $\id_A\otimes \sigma_B$ for the output of the PPT channel $\cM_{B|A}$. 
Despite its normalization, this operator also has the correct $1/M$ overlap with the maximally entangled state. Moreover, it would lead to a hypothesis-testing quantity reminiscent of the coherent information, whose regularized version is part of the formula for quantum channel capacity. However, the resulting optimization is not a semidefinite program, as now we will have to set $M_{B|A}=\phi^{-1}_A\otimes \sigma_B$, leading to $\phi_A^{-1}$ appearing in the objective function along with $\Lambda_{AB}$. Concretely, we obtain
\begin{align}
\frac 1M 
&\geq \min_{\phi_A,\Lambda_{AB}\in \textsc{f}(\cN,\eps)}\,\max_{\sigma_B}\,\tr[\Lambda_{AB}\phi_A^{-1}\otimes \sigma_B]\\
&=\min_{(\phi_A,\Gamma_{AB})\in \textsc{f}'(\cN,\eps)}\,\max_{\sigma_B}\,\tr[\Gamma_{AB}\id_A\otimes \sigma_B],
\end{align}
where $\textsc{f}'(\cN,\eps)$ consists of normalized states $\phi_A$ and $\Gamma_{AB}$ with $0\leq \Gamma_{AB}\leq \id_{AB}$ such that $\tr[\Gamma_{AB}\phi_A^{1/2}\hat N_{B|A}\phi_A^{1/2}]\geq 1-\eps$. 
Furthermore, the symmetry arguments employed in Proposition~\ref{prop:symmetrizefeasible} no longer go through, making the resulting bound difficult to work with. This difficulty is perhaps to be expected, since otherwise symmetrization might lead to a single-letter bound in terms of the coherent information for general channels, which is known to be false. 

The same difficulty applies to formulating an SDP-based bound using any PPT state, not just the output of a PPT channel $\cM_{B|A}$ on the input $\phi_A$ as done in Theorem~\ref{thm:minimaxbound}. This is the approach taken by Tomamichel and Berta to obtain second-order coding rates for simple channels in~\cite{tomamichel_quantum_2015v1}. 
In the notation of this paper, their bound can be expressed as 
\begin{align}
\label{eq:fullpptbound}
\min_{(\phi_A,\Gamma_{AB})\in \textsc{f}'(\cN_{B|A},\eps)}\,\max_{\sigma_{AB}\in\textsc{ppt}} \tr[\Gamma_{AB} \sigma_{AB}]\leq \frac 1M,
\end{align}
with the same $\textsc{f}'$ as in the previous paragraph. (Actually, we have interchanged a minimization over $\Gamma$ with the maximization over $\sigma$, but this is permissible by the minimax theorem.) 
Nevertheless, symmetry arguments do go through in this case, and can be employed to infer that optimal input states $\phi_A$ can be chosen to be invariant under the channel symmetry group. 
For highly symmetric channels such qubit dephasing, depolarization, and erasure, this fixes $\phi_A$ to the mixed state. Then the nonlinearities of the optimization disappear and the resulting bound is identical to the minimax bound~\eqref{eq:minimaxbound}. 
Indeed, for these channels, one could use the results of \S\ref{sec:examples} to more quickly obtain their results regarding second-order coding rates. 
The form of the bound also implies that the minimax bound can be obtained by loosening \eqref{eq:fullpptbound}, restricting the optimization of $\sigma_{AB}$ to states of the form $\sigma_{AB}=\phi_A^{1/2}\hat M_{B|A}\phi_A^{1/2}$ for some PPT channel $\cM_{B|A}$. Equivalently, $\sigma_{AB}$ must be PPT and satisfy $\tr_B[\sigma_{AB}]=\phi_A$.  


\vspace{5mm}
{\bf Acknowledgements.} The author is grateful to David Sutter, Marco Tomamichel, and William Matthews for useful discussions. This work was supported by the Swiss National Science Foundation (through the National Centre of Competence in Research `Quantum Science and Technology') and by the European Research Council (grant No.~258932). 

\printbibliography[heading=bibintoc,title=References]
\end{document}